\documentclass[conference]{IEEEtran}

\usepackage{array,cite}
\usepackage{graphicx,epsfig,subfigure}
\usepackage{amsthm}
\usepackage{amsmath}
\usepackage{mdwmath}
\usepackage{array}
\usepackage{subeqnarray}
\usepackage{stfloats}
\usepackage{algorithm,algorithmic}
\usepackage{xcolor}

\newtheorem{lemma}{\emph{Lemma}}
\newtheorem{theorem}{Theorem}

\IEEEoverridecommandlockouts

\begin{document}

\title{Faster Information Propagation on Highways: a Virtual MIMO Approach}

\author{\authorblockN{Hui~Wu, Zhaoyang~Zhang\IEEEauthorrefmark{2} and Huazi~Zhang\\
Department of Information Science and Electronic Engineering, Zhejiang University, Hangzhou 310027, China\\
Zhejiang Provincial Key Laboratory of Information Network, Zhejiang, China\\
E-mail: 3090102503@zju.edu.cn, \IEEEauthorrefmark{2}ning\_ming@zju.edu.cn, hzhang17@zju.edu.cn}
\thanks{This work was supported in part by National Key Basic Research Program of China (No. 2012CB316104), National Hi-Tech R\&D Program (No.2014AA01A702), Zhejiang Provincial Natural Science Foundation of China (No. LR12F01002), National Natural Science Foundation of China(61371094), National Natural Science Foundation of China (61401388) and the China Postdoctoral Science Foundation Funded Project (No. 2014M551736).}}

\maketitle

\begin{abstract}
In vehicular communications, traffic-related information should be spread over the network as quickly as possible to maintain a safer transportation system. This motivates us to develop more efficient information propagation schemes. In this paper, we propose a novel cluster-based cooperative information forwarding scheme, in which the vehicles opportunistically form virtual antenna arrays to boost one-hop transmission range and therefore accelerate information propagation along the highway. Both closed-form results of the transmission range gain and the improved \emph{Information Propagation Speed} (IPS) are derived and verified by simulations. It is observed that the proposed scheme demonstrates the most significant IPS gain in moderate traffic scenarios, whereas too dense or too sparse vehicle density results in less gain. Moreover, it is also shown that increased mobility offers more contact opportunities and thus facilitates information propagation.
\end{abstract}
\section{Introduction}
\subsection{Motivation}
Vehicular ad hoc networks (VANET) is undergoing extensive study in recent years \cite{model1,model2,model3,add1,add2}. Maintaining a safer transportation system is of top priority. One safety measure is to enable inter-vehicle communication, especially in highways without any fixed road-side infrastructure. Information on traffic-related events, including accident, traffic jam, closed road, etc, needs to be passed on to nearby vehicles in a multi-hop relaying fashion. For instance, if an accident occurs on a highway and induces temporary congestion, it is vital to ensure that all vehicles in this region be informed as quickly as possible, so that some of them can make early detours before getting trapped. Therefore, we are interested in how fast a message can propagate along the route. A new performance metric, \emph{Information Propagation Speed} (IPS) \cite{varyv2}, defined as the distance the information spreads within unit time, is drawing increasing attention.
\subsection{Related Work}
Information spreading in mobile networks has been theoretically studied in \cite{HZ1,HZ2}. Its application in vehicular DTNs has been well studied in literature. Early works \cite{model1,model2,model3,model4,model5} focus on the modeling of vehicular DTNs. In \cite{model4}, the authors studied vehicle traces and concluded that vehicles are very close to being exponentially distributed in highways. Further, measurements in \cite{model5} showed that vehicles traveling in different lanes (e.g. bus lane or heavy truck lane) have different speed distributions.

However, further theoretical analysis of the IPS remained missing until Agarwal et al. first obtained upper and lower bounds for the IPS in a 1D VANET, where vehicles are Poissonly distributed and move at the same, constant speed in opposing directions \cite{upblb1,upblb2}. Their results revealed the impact of vehicle density, indicating a phase transition phenomenon. Under the same setting of \cite{upblb2}, Baccelli et al. derived more fine-grained results of the threshold \cite{exactphasetransi}. The authors complemented their work by taking into account radio communication range variations at the MAC layer, and characterized conditions for the phase transition \cite{varyradiorange}.

Other than traffic density, vehicle speed properties have also been extensively investigated. In \cite{wu}, Wu et al. studied the IPS assuming uniformly distributed, but time-invariant vehicle speeds, and obtained analytical results in extreme sparse and extreme dense cases. \cite{varyv1} showed that the time-variation of vehicle speed will also impose a significant impact on the IPS. When extended to a multi-lane scenario, \cite{varyv2} obtained similar results. However, we also note that, most above-cited works are based on simplified wireless communication models that ignore transmission time. Moreover, each vehicle plays the role of an individual relay, and no inter-vehicle cooperation has been considered.
\subsection{Summary of Contributions}
In this paper, we study the information propagation in a bidirectional highway scenario. We propose a cooperative information forwarding scheme and analyze the corresponding IPS. We show that the proposed scheme yields faster information propagation than conventional non-cooperative schemes.
Our main contributions are summarized below.
\begin{enumerate}
\item{In contrast to the existing highway information propagation models that only consider network layer, we establish a cross-layer framework that features wireless channel dynamics, retransmission delay and enables cooperative transmission. In this sense, our model is more realistic.}
\item{We propose a novel cluster-based information forwarding scheme, in which adjacent vehicles form a distributed antenna array to collaboratively transmit signals. The scheme effectively boosts one-hop transmission range which results in a significant IPS gain. More specifically, we obtain closed-form results of the transmission range gain and the improved IPS.}
\item{We analyze the impact of vehicle speed, traffic density and transmission range on IPS. Additionally, the proposed cooperative scheme demonstrates the most significant IPS gain under moderate traffic, whereas too dense or too sparse vehicle density results in less IPS gain.}
\end{enumerate}

The rest of the paper is organized as follows - In Section \ref{systemmodel}, we present the system model. In Section \ref{rangegainsec}, the transmission range gain from the proposed scheme is evaluated. Closed-form expression of IPS will be derived in Section \ref{IPS}. Simulations results are given in Section \ref{simulation}. The paper is finally concluded in Section \ref{conclusion}.

Throughout the rest of the paper, let $[\textbf{V}]_{M\times N}$ mean that the matrix $\textbf{V}$ is of M rows and N columns, the capital bold style means it is a matrix and the lowercase bold style means it is a vector. $\textbf{E}(\cdot)$ is the mathematical expectation operator and the $\textbf{V}^T$ denotes the transpose.
\section{System Model}
\label{systemmodel}
%\subsection{Cluster-based information forwarding}
%%In this section, we will describe our cross-layer model which includes a physical layer characterization and a virtual MIMO scheme.
\subsection{Cluster-based Opportunistic Forwarding}
We focus on a highway scenario where vehicles are Poissonly distributed with intensity $\lambda$, and travel in opposing directions at a constant speed $v$. A cluster is defined as a maximal sequence of exclusive eastbound (or equivalently, westbound) cars such that any two consecutive cars are within each other's radio range. Without loss of generality, we focus on the information propagation speed in the eastbound lane.

Fig. \ref{propagation pattern} shows the propagation process of a certain packet: it moves toward east from cluster to cluster. We say the packet forwarding is ¡°blocked¡± (see Fig. \ref{propagation pattern_a}) when the transmitter cluster senses that the next eastbound cluster is out of reach. Hence, the packet has to be buffered in the current cluster until the gap is bridged by the opposing traffic (see Fig. \ref{propagation pattern_b}).
\begin{figure}[!h]
\centering
\subfigure[]{\includegraphics[angle=0,scale=0.33]{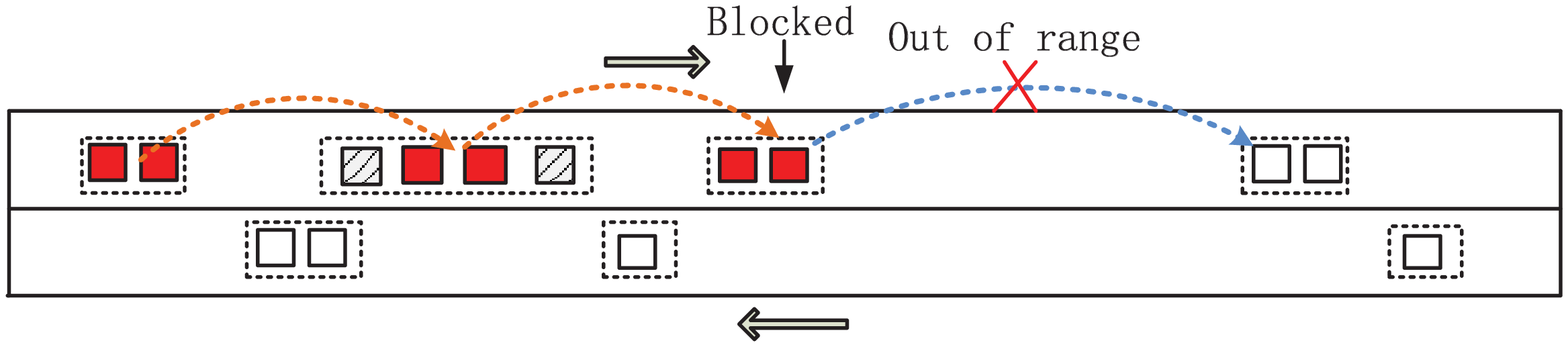}\label{propagation pattern_a}}
\subfigure[]{\includegraphics[angle=0,scale=0.33]{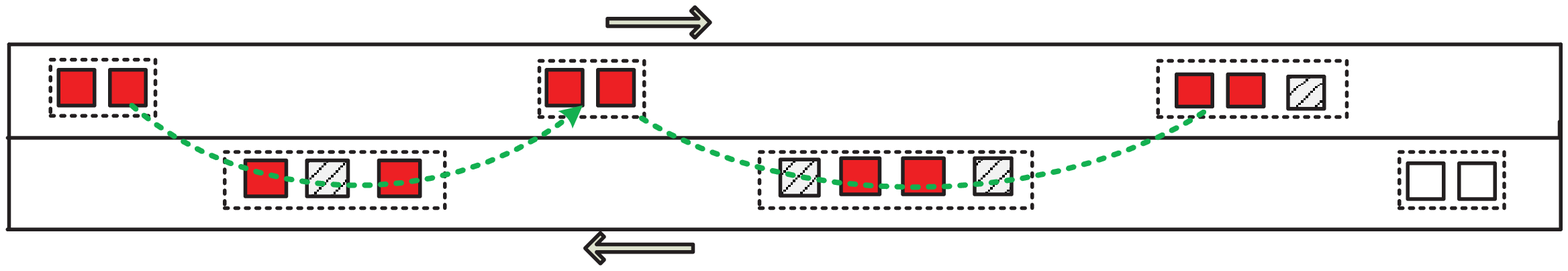}\label{propagation pattern_b}}
\caption{Propagation pattern of a certain packet}
\label{propagation pattern}
\end{figure}
\subsection{Cooperative Transmission Model}
Now we will describe the transmission between two clusters, namely a transmitter cluster (TxC) with $N_t$ ($N_t \geq 2$\footnote{Note $N_t = 1$ is considered as a special case}) vehicles and a receiver cluster (RxC) with $N_r$ vehicles. For ease of analysis, only two randomly-chosen vehicles\footnote{More cooperating nodes only accelerate information spreading.} (denoted by filled squares $T_{x1}$ and $T_{x2}$) in TxC are allowed to transmit.
The cluster-based transmission model is illustrated in Fig. \ref{cluster2cluster}.
\begin{figure}[!h]
\centering
\includegraphics[angle=0,scale=0.5]{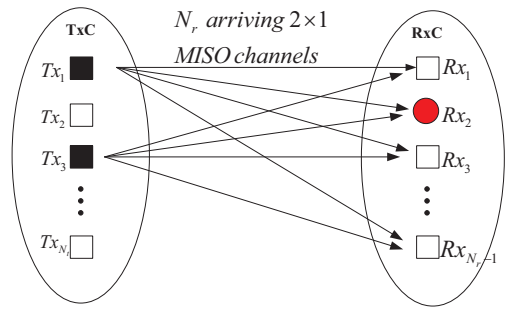}
\caption{Cluster-based cooperative transmission model}
\label{cluster2cluster}
\end{figure}
\subsubsection{Channel Characterization}
The virtual MIMO channel between two transmitters and $N_r$ receivers can be represented as:
\begin{equation}
\textbf{H}=\left[\textbf{D}_1^{\frac{1}{2}}\textbf{s}_1 \quad \textbf{D}_2^{\frac{1}{2}}\textbf{s}_2\right]_{N_r\times2},
\end{equation}
where $\left[\textbf{s}_1\right]_{N_r\times 1}$, $\left[\textbf{s}_2\right]_{N_r\times 1}$ denote the small-scale fading between the RxC and $T_{x1}$ and $T_{x2}$, respectively; $\left[\textbf{D}_1^{\frac{1}{2}}\right]_{N_r\times N_r}$, $\left[\textbf{D}_2^{\frac{1}{2}}\right]_{N_r\times N_r}$ are diagonal matrices accounting for the large scale fading effect.  For mathematical tractability, we approximately assume that the large scale fading of each link is almost the same, thus can be defined as $D^{-\delta}$, here $\delta$ is the path loss exponent with typical value from 2 to 4. Furthermore, we assume that the small-scale fading follows Rician distribution, which is typical for a flat terrain. Thus $\left[\textbf{s}_1\right]_{N_r\times1}$ and $\left[\textbf{s}_2\right]_{N_r \times1}$ can be denoted as follows \cite{rician1}:
\begin{equation}
\begin{aligned}
\textbf{s}_1(\theta_1)&=\sqrt{\frac{K}{K+1}}\textbf{v}(\theta_1)+\sqrt{\frac{1}{K+1}}\textbf{g}_1, \\
\textbf{s}_2(\theta_2)&=\sqrt{\frac{K}{K+1}}\textbf{v}(\theta_2)+\sqrt{\frac{1}{K+1}}\textbf{g}_2
\end{aligned}
\end{equation}
where $K$ is the Rician factor, $\theta_1$ and $\theta_2$ are the Angles of Arrivals (AoAs) of signals from $T_{x1}$ and $T_{x2}$, respectively. Note that in our one-dimensional scenario, $\theta_1{=}\theta_2{=}\theta{=}\pi/2$, thus function $v(\theta)$ is represented as
\begin{equation}
\small
v(\theta){=}\left[1,{...}\exp\left({-}jk\frac{2\pi d_k}{\lambda}\right),{...},\exp\left({-}j(N_r{-}1)\frac{2\pi d_{N_r}}{\lambda}\right) \right]^T,
\end{equation}
and $[g_1]_{N_r\times 1}$, $[g_2]_{N_r \times 1} \sim {\cal{CN}}(\textbf{0},\textbf{I})$ denote the Rician scattering vector.
\subsubsection{Opportunistic Retransmission}
Given the distributed nature of receivers, we adopt a simplified selection diversity \cite{wireless} algorithm at RxC, over the set of $N_r$ MISO channels, i.e., the vehicle with the highest received signal-to-noise ratio (SNR) (denoted by the red-filled circle $R_{x2}$) acts as a coordinator and broadcasts the decoded packet within RxC immediately. We also adopt an ACK-based protocol, so that the coordinator returns a confirmation message (ACK) to TxC on successfully receiving the packet, otherwise the transmitters will keep retransmitting every $\tau$ seconds.
\section{Transmission Range Gain Analysis}
\label{rangegainsec}
In this section, we will analyze the transmission range gain from the proposed cooperative transmission scheme. Intuitively, power gain, provided by jointly transmitting, together with diversity gain, provided by receiving over independent fading channels, can be leveraged to boost the transmission range.

Assume there exists some minimum receive SNR which can be translated to a minimum received power $P_{min}$ for acceptable performance. Providing the target outage probability $P_{out}$, let $r$ represent the maximal one-hop transmission range of a single vehicle with transmit power $P_t$. Based on the flat Rician fading channel model, the received signal on the edge is:
\begin{equation}
y=r^{-\delta/2}\left(\sqrt{\frac{K}{K+1}}+\sqrt{\frac{1}{K+1}}\psi_0\right)x+n_0,
\end{equation}
where $\psi_0\sim$ $\cal{CN}$ $(0,1)$, and $n_0$ denotes the addictive white Gaussian noise with power spectral density $N_0$. Apparently,
\begin{equation}
P_{out}=P(P_0r^{-\delta}|\psi|^2<P_{min}),
\label{Pout}
\end{equation}
where $P_0=\frac{K P_t}{K{+}1}$ denotes the normalized transmit power and $\psi{=}1{+}\sqrt{\frac{1}{K}}\psi_0$.

\begin{lemma}
\label{ncxnormal}
$X$ is an non-central $\chi^2$ distributed random variable with $f$ degrees of freedom and non-central parameter $\lambda$, $(\frac{X}{r})^{\frac{1}{3}}$ is
approximately normally distributed with mean $1-\frac{2}{9}\frac{1+b}{r}$ and variance $\frac{2}{9}\frac{1+b}{r}$, where $r=f+\lambda$ and
$b=\frac{\lambda}{f+\lambda}$.
\end{lemma}
\IEEEproof See \cite{ncx}.

According to Lemma \ref{ncxnormal}, we get
\begin{equation}
\label{approx}
\left(\frac{|\psi|^2}{\sigma^2 \left(2+\frac{1}{\sigma^2}\right)}\right)^{\frac{1}{3}} \sim
{\cal{N}}\left(M_1,V_1\right),
\end{equation}
where $M_1=1-\frac{4\sigma^2\left(\sigma^2+1\right)}{9\left(2\sigma^2+1\right)^2}$, $V_1=\frac{4\sigma^2\left(\sigma^2+1\right)}{9\left(2\sigma^2+1\right)^2}$ and $\sigma=\sqrt{\frac{1}{K}}$.

By combining \eqref{Pout} and \eqref{approx}, we easily obtain
\begin{equation}
\label{r}
r=\left[\frac{P_0\sigma^2\left(2+\frac{1}{\sigma^2}\right)
\left(\sqrt{V_1}\Phi^{-1}\left(P_{out}\right)+M_1\right)^3}
{P_{min}}\right]^{\frac{1}{\delta}}.
\end{equation}

In terms of the cooperative scheme, with each vehicle transmitting the same signal with identical transmit power $P_t$, the received signal can be represented as:
\begin{equation}
\textbf{y}=\left[\textbf{D}_1^{\frac{1}{2}}\textbf{s}_1 \quad\textbf{D}_2^{\frac{1}{2}}\textbf{s}_2\right]
\left[\begin{array}{l}
x\\
x\\
\end{array} \right]+\textbf{n},
\end{equation}
where $\textbf{E}(x^2)=P_t$ and $[\textbf{n}]_{N_r\times 2} \sim {\cal{CN}}(0,\sigma_N^2 \textbf{I})$ is the complex white Gaussian noise vector.

We assume perfect phase compensation at RxC, so the output SNR, based on the selection diversity algorithm, is:
\begin{equation}
\gamma^{SC}=\frac{P_t}{N_0} \mathop{\max}\limits_{k}|\textbf{h}_k|^2,
\end{equation}
where $\textbf{h}_k$ denotes the $k$th row of $\textbf{H}$.

Let $R$ be the expanded transmission range, referred as MIMO range for brevity. Given the same target outage probability on the edge, thus we have
\begin{equation}
\begin{aligned}
P_{out}&=P\left(\gamma^{SC}<\frac{P_{min}}{N_0}\right)\\
       &=\prod\limits_{k{=}1}^{N_r} {P\left({P_0}{R^{-\delta}}\left(|\psi_{k1}|^2+|\psi_{k2}|^2\right){<}P_{min}\right)} \\
       &=\left[P\left(|\psi_{k1}|^2+|\psi_{k2}|^2{<}{\frac{P_{min}}{P_0 R^{-\delta}}}\right)\right]^{N_r}
\end{aligned}
\end{equation}
Note that the second and the third equality follow from the cumulative density function of the maximum of i.i.d random variables.
Similarly, we have $\frac{|\phi_{k1}|^2+|\phi_{k2}|^2}{\sigma^2} \sim \chi^{'2}_4(\frac{2}{\sigma^2})$, which gives:
\begin{equation}
R=\left[\frac{2P_0\sigma^2\left(2+\frac{1}{\sigma^2}\right)\left(\sqrt{V_2}\Phi^{-1}\left(P_{out}^{\frac{1}{N_r}}\right)+M_2\right)^3}{P_{min}}\right]^{\frac{1}{\delta}},
\end{equation}
where $M_2=1-\frac{2\sigma^2\left(\sigma^2+1\right)}{9\left(2\sigma^2+1\right)^2}$ and
$V_2=\frac{2\sigma^2\left(\sigma^2+1\right)}{9\left(2\sigma^2+1\right)^2}$.
By taking the ratio of $R$ and $r$, the transmission range gain is:
\begin{equation}
\label{rangegain}
g=[\frac{2(\sqrt{V_2}\Phi^{-1}(P_{out}^{\frac{1}{N_r}})+M_2)^3}{(\sqrt{V_1}\Phi^{-1}(P_{out})+M_1)^3}]^{\frac{1}{\delta}}.
\end{equation}

Since $P_{out}$, $\sigma$ are predetermined constants, $N_r$ is left to exclusively affect $g$. Equivalently, we can write $g$ as a function of $N_r$:
\begin{equation}
\label{g}
g \buildrel \Delta \over = F(N_r).
\end{equation}

Fig. \ref{rangegainfig} shows the relationship between $g$ and $N_r$, with $N_r$ plotted in log-fashion. It can be seen that the range gain grows logarithmically with the number of cooperating vehicles (or cluster size). This implies a few additional cooperating vehicles would be sufficient, and too many cooperators would only incur excessive overhead.
\begin{figure}[!ht]
\centering
\includegraphics[angle=0,scale=0.33]{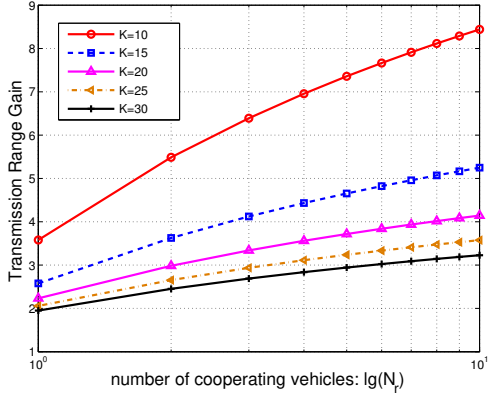}
\caption{Transmission range gain as number of cooperating vehicles increases}
\label{rangegainfig}
\end{figure}
\section{Information Propagation Speed Analysis}
\label{IPS}
The technical building blocks for analyzing IPS is organized as follows. \ref{numbercluster} sutdies the distribution of cluster size, i.e., the number of vehicles in a cluster. \ref{roadlengthgapsub} analyzes the average road length ahead until the gap is bridged by the next cluster. In \ref{unbrigapsub}, we investigate the distribution of an unbridged gap length. Based on the results in \ref{numbercluster} to \ref{unbrigapsub}, the expectation of blocking time is obtained in \ref{subsecT}. Meanwhile, the expected distance that a piece of information can propagate after the block is given in \ref{Dsubsec}. Having obtained these key elements, we get the final expression of the IPS in \ref{IPSsubsec}.
\subsection{Number of Vehicles Within a Cluster}
\label{numbercluster}
\begin{lemma}
\label{vehiclenumberdistri}
The probability mass function (pmf) of the number of vehicles within a cluster is given by:
\begin{equation}
\label{PNk}
P_N(k)=e^{-\lambda r}(1-e^{-\lambda r})^{k-1} ~~~~~ (k=1,2,3\ldots)
\end{equation}
and its cumulative distribution function (CDF) of is:
\begin{equation}
F_N(n)=P(N \leq n)=1-(1-e^{-\lambda r})^n,
\end{equation}
where $\lambda$ is the poisson intensity.
\end{lemma}
\IEEEproof
This is a straightforward result in probability theory.
\subsection{Road Length to Bridge a Gap}
\label{roadlengthgapsub}
As aforementioned, when a packet is blocked by a gap of length $x$, it has to wait for help from the opposing traffic. Let $B(x)$ be the distance to the first westbound cluster ahead which is capable of bridging the gap (starting from an arbitrary cluster). Since the MIMO range only associates with the number of potential receivers, a westbound cluster of at least $n_0$ vehicles is needed, and $n_0$ is given as:
\begin{equation}
\label{n_0}
n_0 = \lfloor F^{-1}(x/r-F(N_r)) \rfloor \buildrel \Delta \over =G(N_r,x).
\end{equation}

Fig. \ref{B(x)} depicts a gap of length $x$, and the distance to the available westbound cluster.
\begin{figure}[!ht]
\centering
\subfigure[]{\includegraphics[angle=0,scale=0.37]{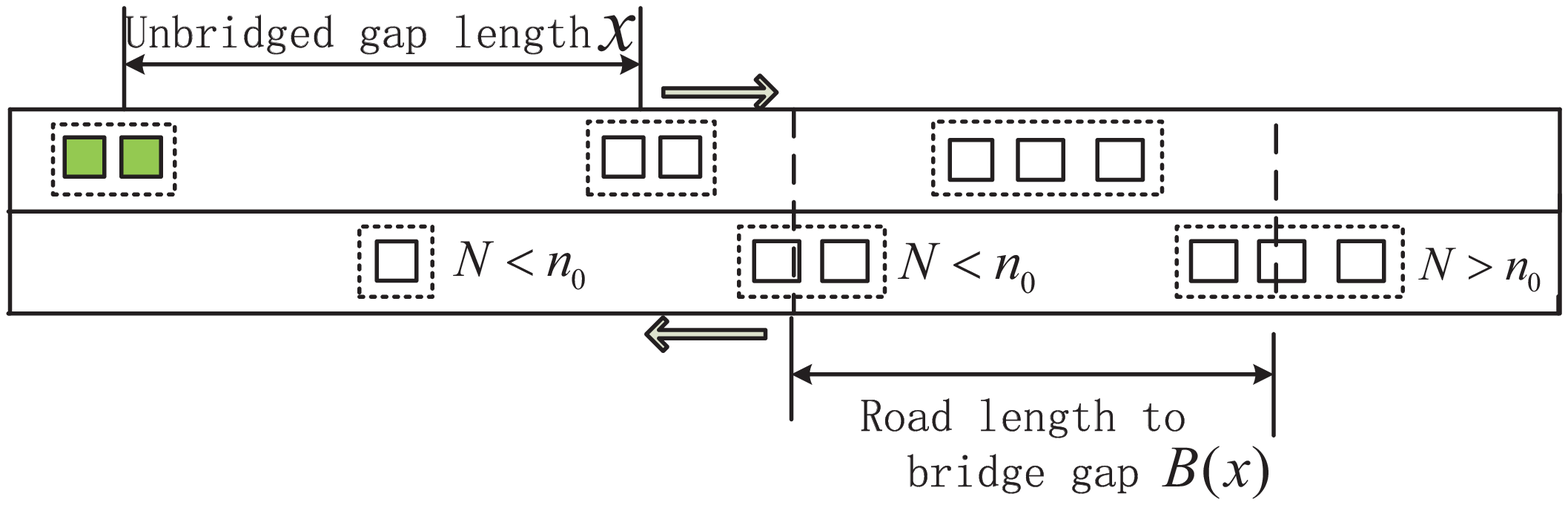}\label{smaller cluters}}
\subfigure[]{\includegraphics[angle=0,scale=0.29]{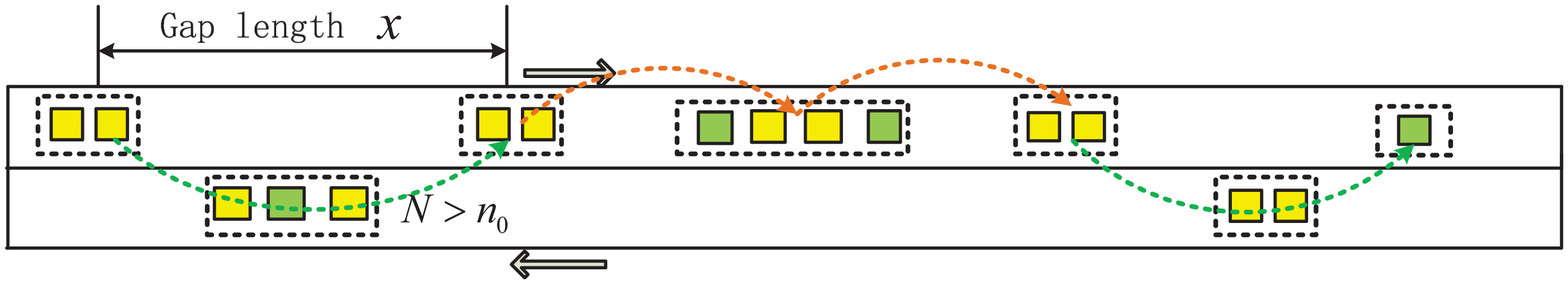}\label{cluster with more than n0 cars}}
\caption{Illustration of the $B(x)$ until a gap $x$ is bridged. (a) Smaller clusters cannot bridged the gap. (b) Until a westbound cluster with more than $n_0$ vehicles is encountered.}
\label{B(x)}
\end{figure}

Denote $\beta(\theta,x)$ the general Laplace transform of $B(x)$, i.e., $\beta(\theta,x)=\textbf{E}(\exp(-\theta B(x)))$.
\begin{lemma}
\label{roadlengthB(x)}
\begin{equation}
\label{theta}
\beta(\theta,x)=\frac{P(N > n_0)}{1-P(N \leq n_0)\frac{\lambda}{\lambda+\theta}},
\end{equation}
\begin{equation}
\label{EBx}
E(B(x)){=}\sum\limits_{k{=}1}^{k_0} \frac{P_N(k)}{\lambda}\left((1{-}e^{-\lambda r})^{{-}G(k,x)}{-}1\right),
\end{equation}
where $G(\cdot,\cdot)$ is given by (\ref{n_0}), and $k_0=\lfloor F^{-1}(x/r)\rfloor$.
\end{lemma}
\IEEEproof
\eqref{theta} follows from renewal theory where the Laplace transform of inter-cluster distance equals $\frac{\lambda}{\lambda+\theta}$. Hence, the average is
\begin{equation}
\begin{aligned}
E\left(B(x)\right)&{=}\sum\limits_{k{=}1}^{k_0}P(N_r{=}k)E(B(x)|N_r{=}k)   \\
   &{=}\sum\limits_{k{=}1}^{k_0}P_N(k)\left(-\frac{\partial}{\partial\theta}\beta(\theta,x)|_{\theta=0}\right)\\
   &{=}\sum\limits_{k{=}1}^{k_0}\frac {P_N(k)}{\lambda}\left(\left(1{-}e^{-\lambda r}\right)^{{-}G(k,x)}{-}1\right)
\end{aligned}
\end{equation}
\subsection{Unbridged Gap Distribution}
\label{unbrigapsub}
Fig. \ref{unbridgedgap} illustrates the two cases where the gap cannot be bridged. Here, we only consider the former case which accounts for the lower bound of IPS.
\begin{figure}[!ht]
\centering
\subfigure[]{\includegraphics[angle=0,scale=0.37]{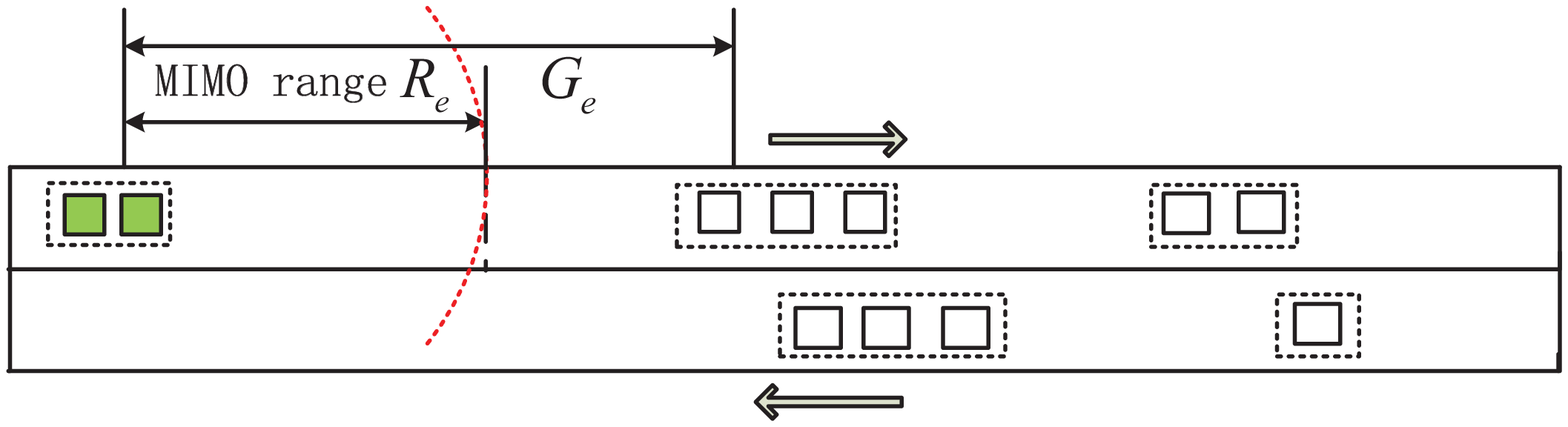}\label{unbridgedcase1}}
\subfigure[]{\includegraphics[angle=0,scale=0.43]{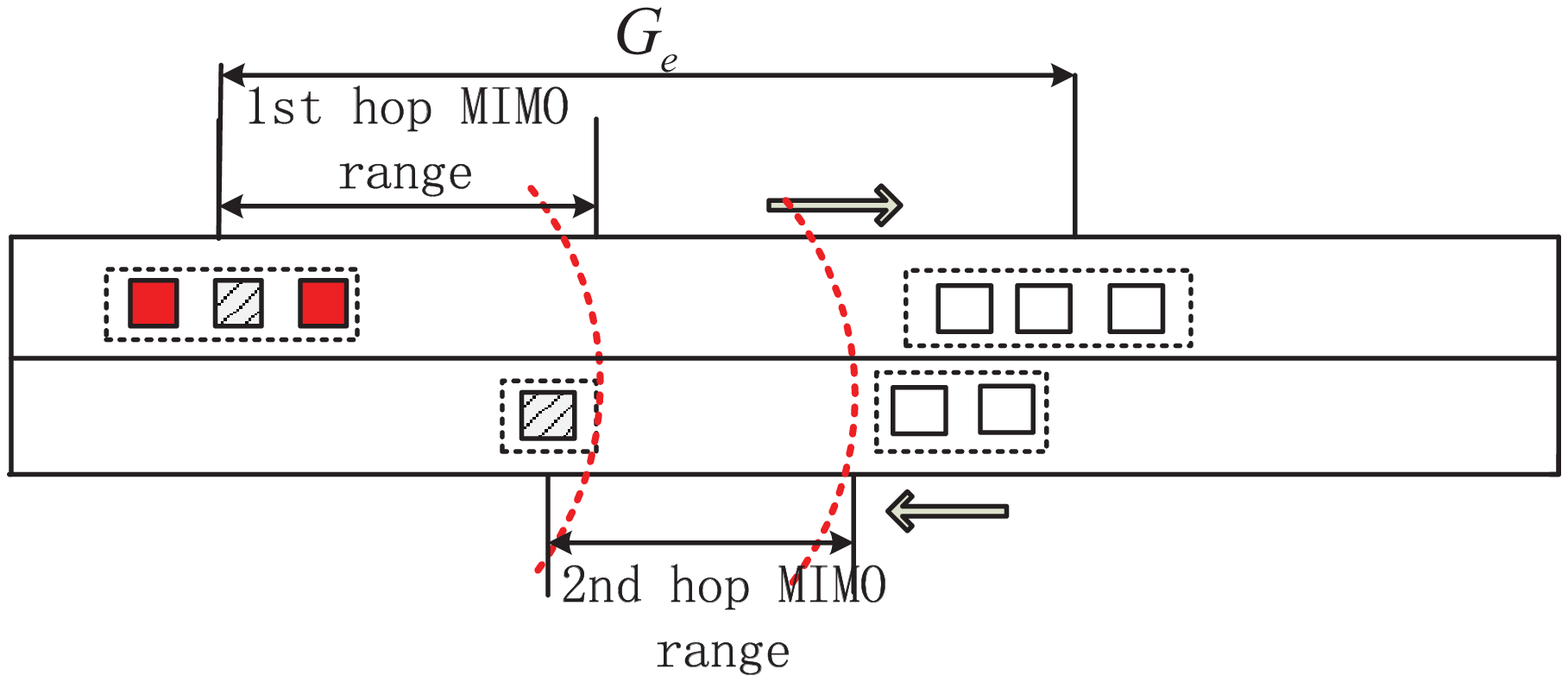}\label{unbridgedcase2}}
\caption{Illustration of an unbridged gap: (a) the next eastbound car is out of MIMO range, and there exists no westbound car within the gap; (b) a westbound
cluster is encountered, yet it is not large enough to aid in bridging.}
\label{unbridgedgap}
\end{figure}
\begin{lemma}
The expected unbridged gap is:
\begin{equation}
\label{EG_e}
\textbf{E}(G_e)=\sum \limits_{k{=}1}^{\infty}\frac {(1{-}e^{{-}\lambda r})^{k{-}1}\int_{rF(k)}^{\infty} x\lambda e^{-\lambda x}dx}{e^{-\lambda r(F(k)-1)}},
\end{equation}
where $F(\cdot)$ is given by (\ref{rangegain}) and (\ref{g}).
\end{lemma}
\begin{lemma}
The $\emph{pdf}$ of $G_e$ is:
\begin{equation}
\label{pex}
p_e(x)=\sum \limits_{k=1}^{\infty}{P(N=k)p_e\left(x|x>rF(k)\right)},
\end{equation}
\begin{eqnarray}
p_e\left(x|x>rF(k)\right)=\left\{
\begin{aligned}
& \frac{\lambda e^{-\lambda x}}{e^{-\lambda rF(k)}}, &\quad x\geq rF(k)  \\
& 0,  &\quad x< rF(k)\\
\end{aligned}
\right.
\end{eqnarray}
\end{lemma}
\IEEEproof
\eqref{EG_e} and \eqref{pex} come from the exponential property of inter-cluster distance and the fact that MIMO range follows the same distribution as $N_r$.
\subsection{Distribution of Blocking Time $T_w$}
\label{subsecT}
\begin{lemma}
\label{Twlemma}
Let $T_w$ denote the blocking time until connectivity provided by the opposing traffic is available. The expected value of $T_w$ satisfies:
\begin{equation}
2v\textbf{E}(T_w)=\textbf{E}(G_e)+\frac {1}{\lambda}+\int_{0}^{\infty}{\textbf{E}(B(x))p_e(x)}dx.
\end{equation}
\end{lemma}
\begin{proof}
 As depicted in Fig. \ref{T}, the total distance to be traversed by two opposing traffic until the gap is bridged, equals the distance to the first westbound cluster $G_e+I_w$ plus the road length to bridge the
 gap $B(x)$ (starting from an arbitrary cluster). Apparently, $2vT_w=G_e+I_w+B(x)$. The proof is completed by taking the expectations, and averaging on all possible gap length $x$.
 \begin{figure}[!ht]
\centering
\includegraphics[angle=0,scale=0.37]{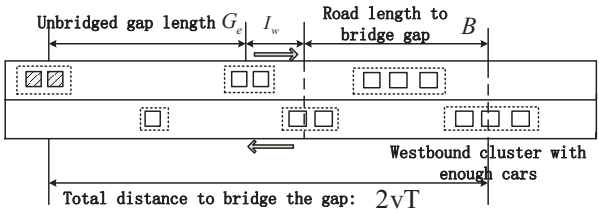}
\caption{Blocking Time $T$: the total distance to bridged the gap is $G_e+I_w+B(x)$}
\label{T}
\end{figure}
\end{proof}
\subsection{Distance $D$ Traversed After the Blocking}
\label{Dsubsec}
Let $C_e$ denote the relative eastbound distance traversed beyond the first gap, until the next block is encountered. As shown in Fig. \ref{D}, we have
\begin{equation}
\label{E(D)}
\textbf{E}(D)=\textbf{E}(G_e)+\textbf{E}(C_e).
\end{equation}
\begin{lemma}
\label{Dlemma}
The average value of $D$ is given by:
\begin{equation}
\label{E(D)lemma}
\textbf{E}(D)=\frac{1}{\lambda(1-P_b)}
\end{equation}
where $P_b$ indicates the probability that a gap can be bridged, and is given as:
\begin{equation}
\label{Pb}
\begin{aligned}
P_b&{=}\frac{1}{2}\sum\limits_{n_e{=}1}^{\infty}{\left(1-e^{-\lambda rF(n_e)}\right)P_N(n_e)} \\
   &{+}\frac{1}{2}\sum\limits_{n_w{=}1}^{\infty}\sum\limits_{n_e{=}1}^{\infty}{\left(1{-}e^{-\lambda r\left(F(n_w){+}F(n_e)\right)}\right)P_N(n_w)P_N(n_e)}.
\end{aligned}
\end{equation}
\end{lemma}
\begin{figure}[!h]
\centering
\includegraphics[angle=0,scale=0.3]{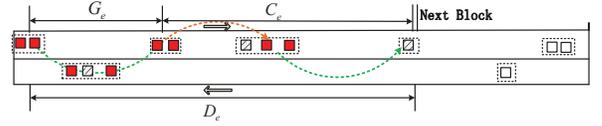}
\caption{Total distance traversed when a bridge is created $D=G_e+C_e$}
\label{D}
\end{figure}
\IEEEproof
identity \eqref{Pb} comes from the exponential nature of a gap, i.e., inter-cluster distance.

Let $\bar{G}_e$ be the length of a bridged gap, we can obtain
\begin{equation}
\label{unconditional}
(1-P_b)\textbf{E}(G_e)+P_b\textbf{E}(\bar{G}_e)=\frac{1}{\lambda},
\end{equation}
\begin{equation}
\label{C_e}
\textbf{E}(C_e)=P_b\left(\textbf{E}(\bar{G}_e)+\textbf{E}(C_e)\right).
\end{equation}

We complete the proof by substituting \eqref{Pb}-\eqref{C_e} into \eqref{E(D)}.
\subsection{The IPS derivation}
\label{IPSsubsec}
As discussed in Section \ref{systemmodel}, the entire information propagation process can be abstracted to a renewal reward process where each cycle is comprised of a \emph{waiting} state and a \emph{forwarding} state. Let sequences ${T_i}$ and ${D_i}$ denote the duration of waiting state and the distance traversed in forwarding state, respectively. The tuples in the sequence $(T_i,D_i)$ are independent and identically distributed, thanks to the Poisson nature of vehicle traffic. Thus, the average referential information propagation speed $v_p$ (with respect to the eastbound cars) is defined as:
\begin{equation}
\label{Vp}
v_p=\frac{\textbf{E}(D)}{\textbf{E}(T)}.
\end{equation}
Note that $T_i$ includes the blocking time $T_{wi}$ and the transmission time $T_{ti}$, i.e.,
\begin{equation}
\label{ET}
\textbf{E}(T)=\textbf{E}(T_w)+\textbf{E}(T_t)
\end{equation}

For $T_t$, it is reasonable to assume that transmission via either a one-hop or two-hop MIMO transmission approximately occurs with the same probability, after a sufficient long observation. Denote by $P_o$ the outage probability of a single hop, we have
\begin{equation}
\begin{aligned}
\label{ETt}
\textbf{E}(T_t)&\approx \frac{3}{2}\textbf{E}\left(\frac{\tau}{1-P_o}\right) \\
               &\approx \frac{3\tau}{2}\frac{1}{1-\textbf{E}(P_o)} \\
               & =\frac{3\tau}{2}\frac{1}{1{-}\sum \limits_{k{=}1}^{\infty}P_N(k)\textbf{E}\left(P_o|_{N_r{=}k}\right)},
\end{aligned}
\end{equation}
where the expression of $P_o$ is the same as \eqref{Pout}, except that $R$ should be replaced with a random exponential variable $I$, and the expectation is taken over $I$. Also note the second approximation is due to the fact that $P_o$ is usually small enough so that $\textbf{E}(1-P_o) \approx \textbf{E}(\frac{1}{1-P_o})$ with $1-P_o$ approaches 1.

\begin{theorem}
The IPS is given by:
\begin{equation}
v_p=\frac{\textbf{E}(D)}{\textbf{E}(T_w)+\textbf{E}(T_t)}
\end{equation}
where $\textbf{E}(D)$ is defined by Lemma \ref{Dlemma}, $\textbf{E}(T_t)$ is given by \eqref{ETt} and $\textbf{E}(T_w)$ is defined by Lemma \ref{Twlemma}.
\end{theorem}

Table \ref{IPStable} gives an overview of our main analytical results.
\begin{table}[!h]
\caption{Summary of main analytical results}\label{IPStable}
  \renewcommand{\arraystretch}{1.5}
  \centering
  \begin{tabular}{|c|c|c|}
  \hline
  \multicolumn{3}{|c|}{IPS $v_p{:} Eq. \eqref{Vp}$} \\
  \hline
  Distance ${\bf{E}}(D){:}$             & \multicolumn{2}{c|}{Total time ${\bf{E}}(T){:} Eq. \eqref{ET}$} \\ \cline{2-3}
  $Eq. \eqref{E(D)lemma}$  & Waiting time ${\bf{E}}(T_w){:}$   &  Transmission time ${\bf{E}}(T_t){:}$   \\
  $Eq. \eqref{Pb}$         & $Eq. \eqref{B(x)},\eqref{pex},\eqref{EG_e}$ & $Eq. \eqref{ETt}$ \\ \hline
  \end{tabular}
\end{table}
\section{Simulation Results}
\label{simulation}
In this section, we conduct information propagation experiments to verify the correctness and accuracy of the derived theoretical results. The simulation follow precisely the model described in Section \ref{systemmodel}. We measure the IPS by selecting a sufficient remote source-destination pair, taking the ratio of the propagation distance over the corresponding delay, and averaging over multiple iterations of randomly generated traffic.
\begin{figure*}[!t]
\normalsize
\subfigure[r=20(m)]{\includegraphics[height=1.7in,width=2.4in]{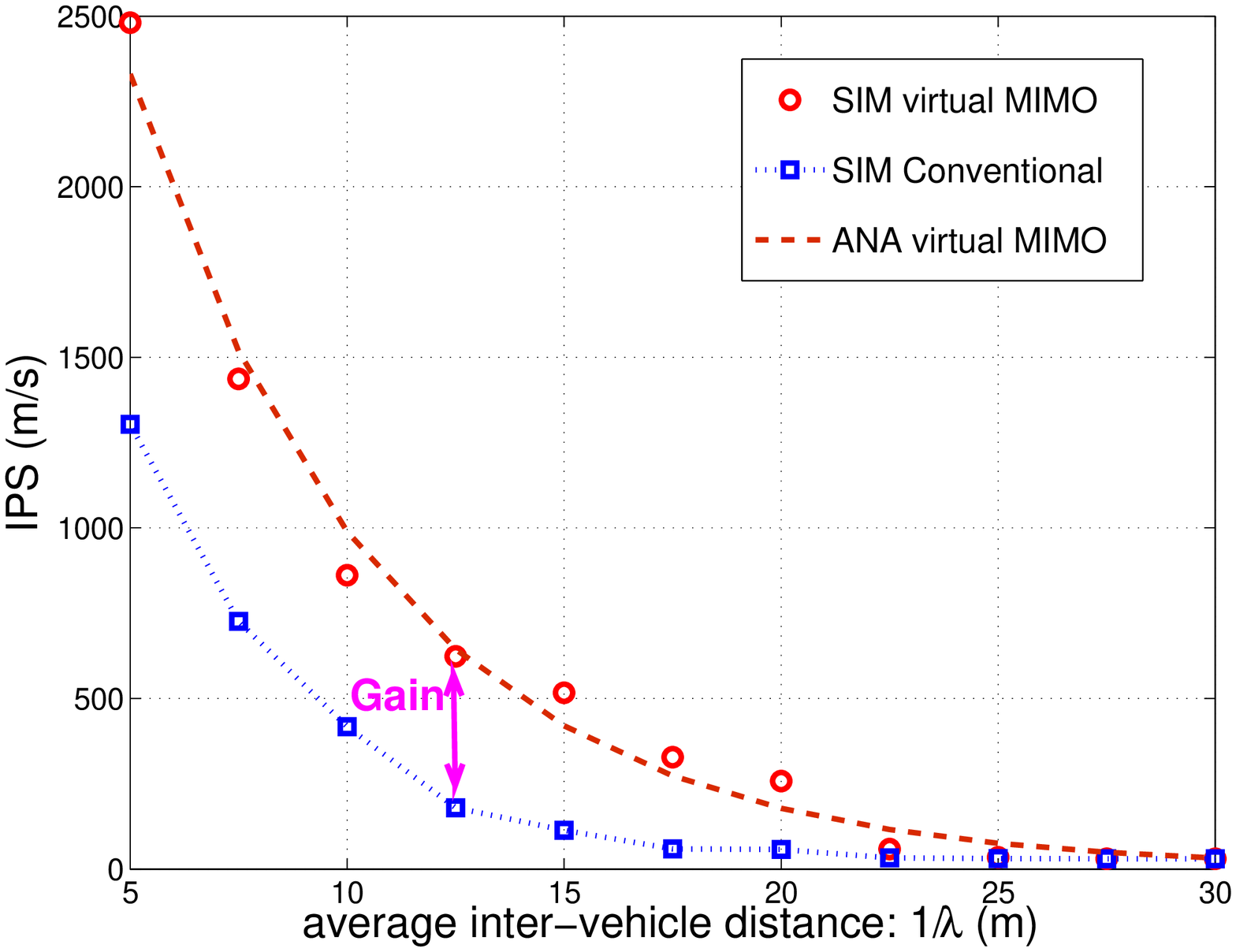}\label{vprop_lambda_r_20}}
\subfigure[r=25(m)]{\includegraphics[height=1.7in,width=2.4in]{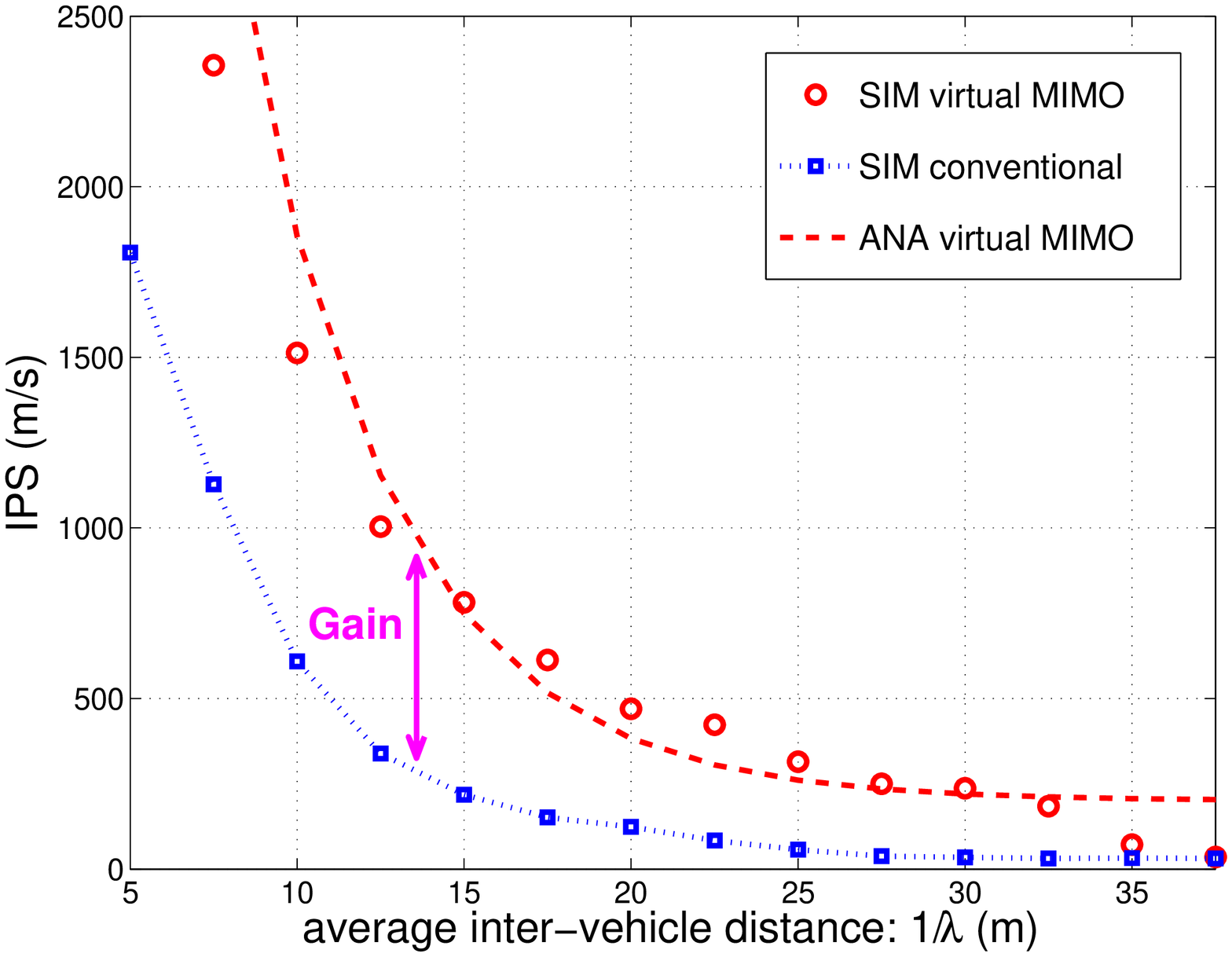}\label{vprop_lambda_r_25}}
\subfigure[r=30(m)]{\includegraphics[height=1.7in,width=2.4in]{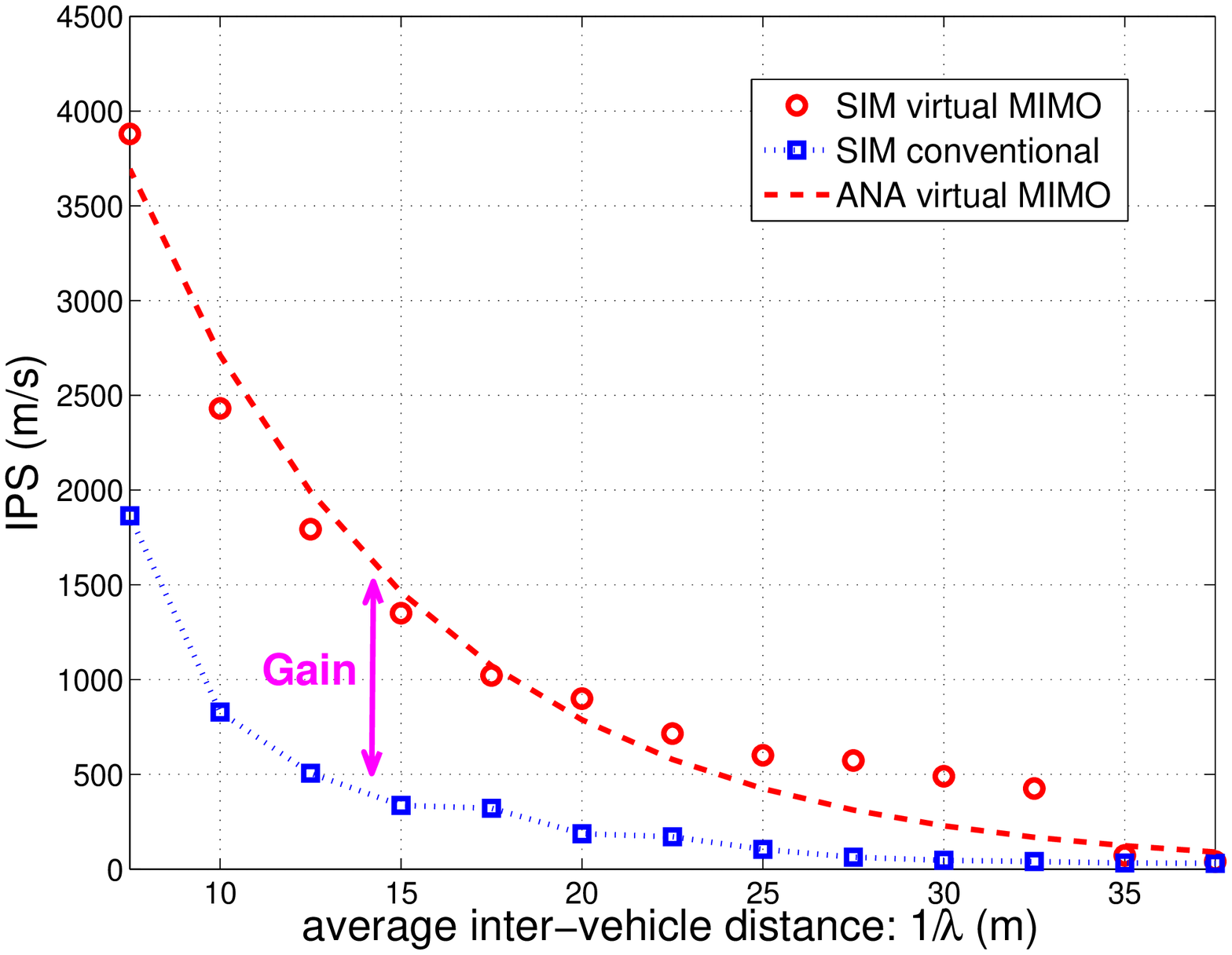}\label{vprop_lambda_r_30}}
\caption{The IPS with different vehicle densities}
\label{IPSvslambda}
\end{figure*}
\begin{figure*}[!t]
\normalsize
\subfigure[r=20(m)]{\includegraphics[height=1.7in,width=2.4in]{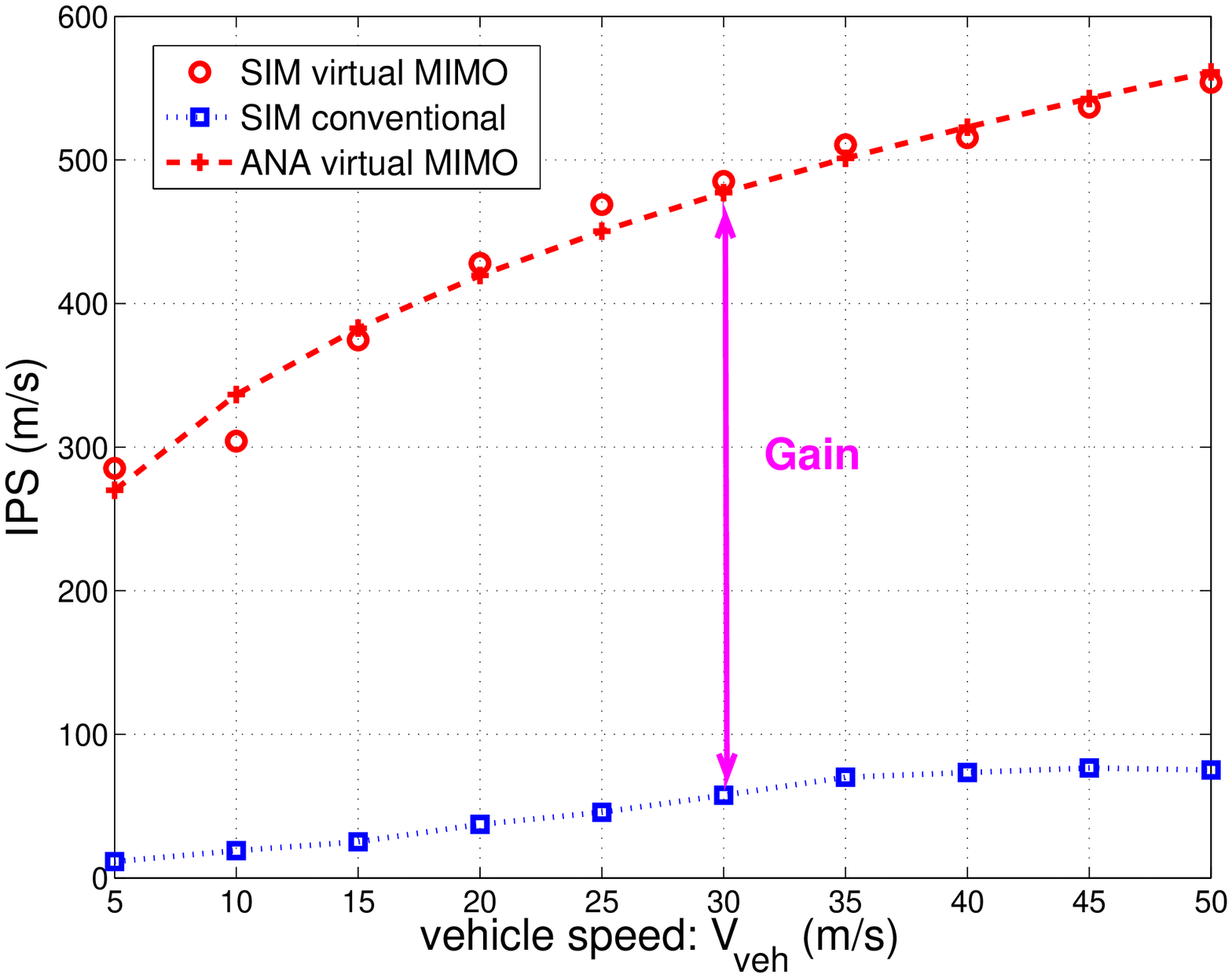}\label{vprop_v_r_20}}
\subfigure[r=25(m)]{\includegraphics[height=1.7in,width=2.4in]{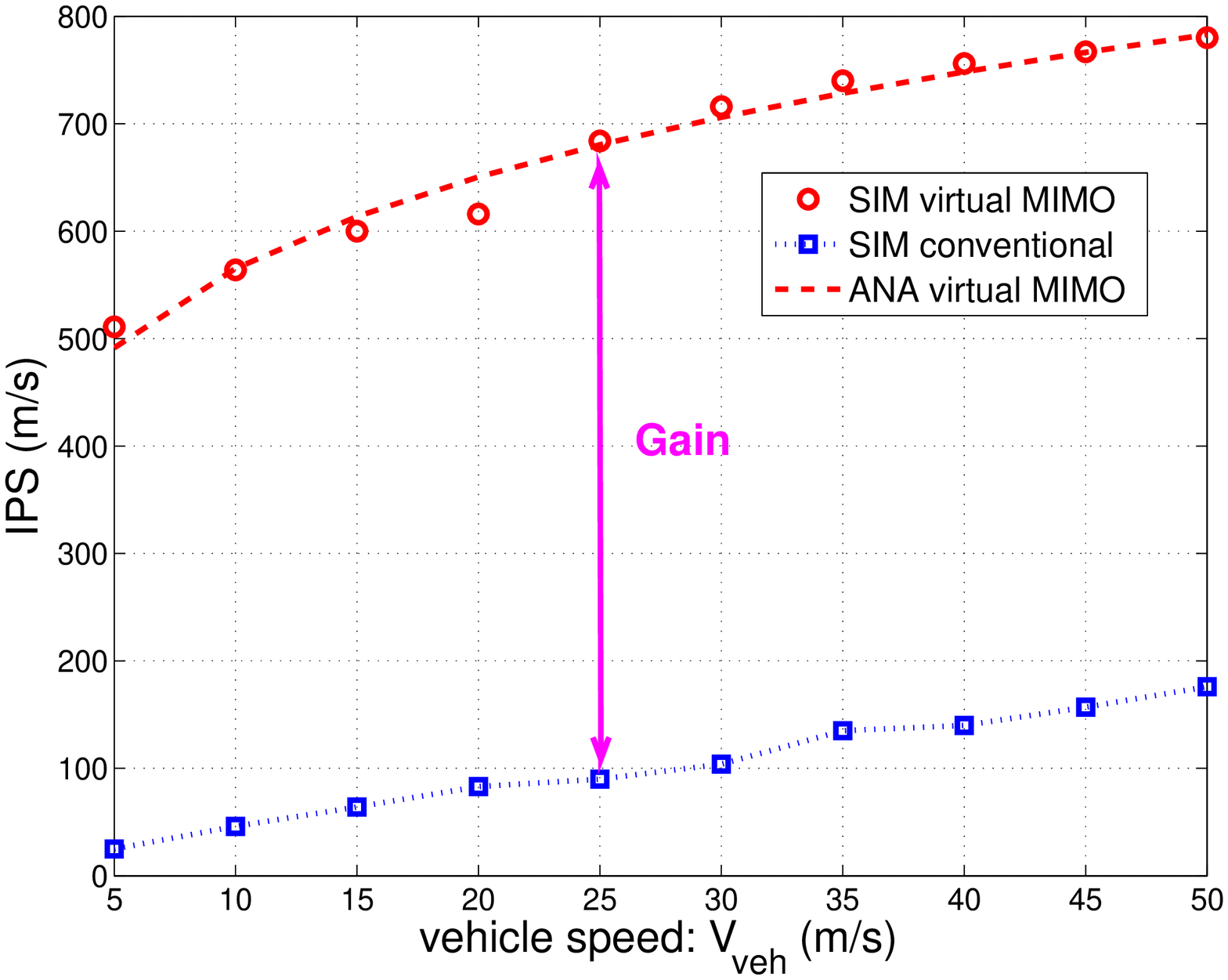}\label{vprop_v_r_25}}
\subfigure[r=30(m)]{\includegraphics[height=1.7in,width=2.4in]{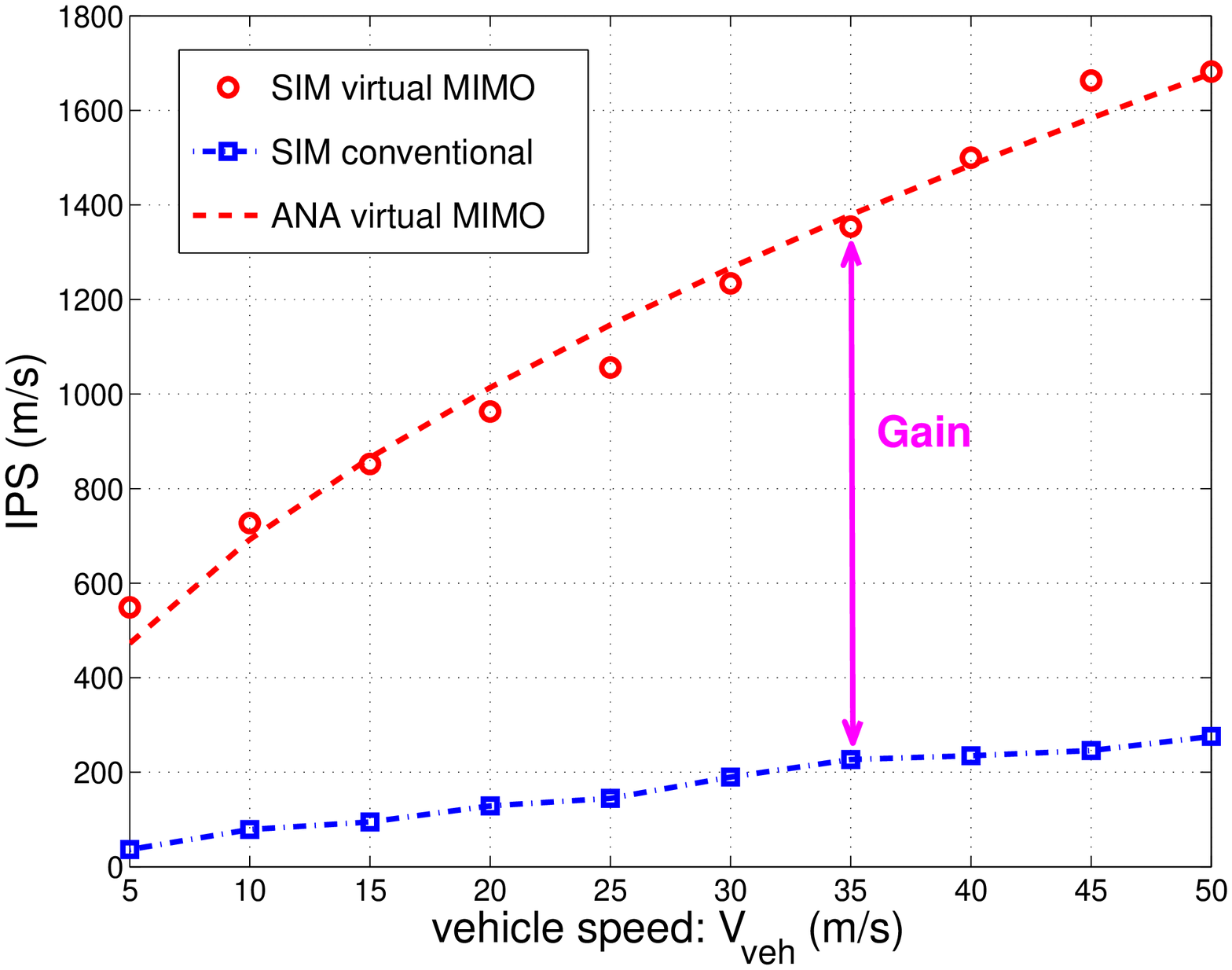}\label{vprop_v_r_30}}
\caption{The IPS with different vehicle speeds}
\label{IPSvsv}
\end{figure*}
\begin{figure}[!h]
\centering
\includegraphics[angle=0,scale=0.35]{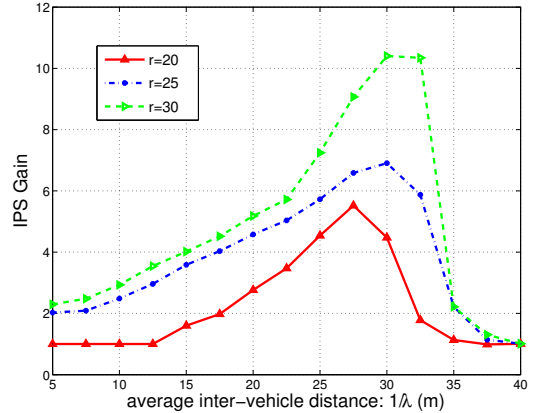}
\caption{The IPS gain over traditional propagation scheme with different vehicle densities}
\label{IPSgainsim}
\end{figure}

Fig. \ref{IPSvslambda} shows the comparisons of IPS between the proposed cooperative forwarding scheme and conventional non-cooperative scheme under different single-vehicle transmission ranges. In each setting, the average inter-vehicle distance varies from 5m to 40m as vehicles move at 30 m/s. We observe that the proposed scheme exhibits significant advantage over the conventional one. For instance, in Fig. \ref{vprop_lambda_r_25}, when the vehicle density is 50 vehicles/km, i.e., the average inter-vehicle distance is 20m, the IPS is raised from 120m/s to about 500m/s. Under the same vehicle density with larger single-vehicle transmission range (see Fig. \ref{vprop_lambda_r_30}), the raise is even more obvious, from 200m/s to about 1000m/s. The reason is that larger $r$ accounts for larger clusters, thus MIMO range is further increased with more cooperating vehicles and vise versa.

Interestingly, performance gain introduced by our scheme is the most remarkable under moderate traffic, whereas too dense or too sparse traffic makes the IPS gain less (see Fig. \ref{IPSgainsim}). Intuitively, in the extreme sparse case, the network is always disconnected and no cooperation opportunity can be exploited, while in the extreme dense case, the network is fully connected so that our cluster-based cooperative forwarding scheme cannot further improve the connectivity of the highway network.

Fig. \ref{IPSvsv} also displays the IPS comparisons for increasing vehicle speed. Again, the proposed scheme is superior. It is also shown that faster traffic actually accelerates information propagation (e.g. a huge IPS enhancement from 100m/s to around 700m/s when the vehicles move at 25m/s in Fig. \ref{vprop_v_r_25}), since increased mobility brings more dynamics to the network topology and creates more opportunistic contacts. Moreover, in both experiments, there is close match between simulation results and theoretical results.

Furthermore, in both experiments, there is close match between simulation results and theoretical results, and IPS also increases with $r$. This is because as $r$ becomes larger, clusters becomes larger, thus MIMO range is further increased.
\section{Conclusion}
\label{conclusion}
In this paper, we study the information propagation speed in the bidirectional highway scenario, based on a novel cross-layer analytical framework. We proposed a cluster-based information forwarding scheme, in which adjacent vehicles form a distributed antenna array to cooperatively and opportunistically transmit signals. It is found that the scheme can effectively boost one-hop transmission range and exhibits significant advantage in IPS. Both closed-form results of the transmission range gain and the improved IPS are derived. Interestingly, it is also shown that the magnitude of IPS gain depends upon traffic density, and increased mobility offers more opportunistic contacts thus aids in information propagation. In our future work, we intend to further generalize the model by engaging all the vehicles within a cluster in cooperation and taking into account different vehicle densities in opposing directions.

%\section*{Acknowledgement}
%This work was supported in part by National Key Basic Research Program of China (No. 2012CB316104), National Hi-Tech R\&D Program (No.2014AA01A702), Zhejiang Provincial Natural Science Foundation of China (No. LR12F01002), National Natural Science Foundation of China(61371094), National Natural Science Foundation of China (61401388) and the China Postdoctoral Science Foundation Funded Project (No. 2014M551736).

\bibliographystyle{IEEEtran}

\end{document}